\documentclass[preprint,12pt,authoryear]{elsarticle}
\usepackage[margin=1in]{geometry}
\usepackage{graphicx} 
\usepackage{amsmath}
\usepackage{amsfonts}
\usepackage{bm}
\usepackage{natbib}
\usepackage{amsthm}
\usepackage{amssymb}
\usepackage{mathtools}
\usepackage{enumitem}
\usepackage{stmaryrd}
\usepackage{hyperref}

\journal{HAL}

\newtheorem{theorem}{Theorem}
\newtheorem{definition}{Definition}
\newtheorem{lemma}{Lemma}
\newtheorem{iidassumption}{Assumption A\ignorespaces}
\newtheorem{assumption}{Assumption B\ignorespaces}

\newcommand{\iid}{i.i.d.}

\newcommand{\bM}{M(\psi)}
\newcommand{\bv}{v(\psi)}

\mathtoolsset{showonlyrefs}

\DeclareMathOperator*{\argmin}{arg\,min}

\date{October 2023}

\begin{document}

\begin{frontmatter}

\author[label1,label2]{Herbert Susmann\corref{cor1}}\ead{herbert.susmann@dauphine.psl.eu}
\cortext[cor1]{Corresponding author.}
\author[label3]{Antoine Chambaz}
\affiliation[label1]{organization={CEREMADE (UMR 7534), Université Paris-Dauphine PSL},
            addressline={Place du Maréchal de Lattre de Tassigny},
            city={Paris},
            postcode={75016},
            country={France}}
\affiliation[label3]{organization={MAP5 (UMR 8145), Université Paris Cité},
            postcode={75006},
            city={Paris},
            country={France}}

\title{Quantile Super Learning for independent and online settings with application to solar power forecasting}

\begin{abstract}
Estimating quantiles of an outcome conditional on covariates is of fundamental interest in statistics with broad application in probabilistic prediction and forecasting. We propose an ensemble method for conditional quantile estimation, Quantile Super Learning, that combines predictions from multiple candidate algorithms based on their empirical performance measured with respect to a cross-validated empirical risk of the quantile loss function. We present theoretical guarantees for both \iid\ and online data scenarios. The performance of our approach for quantile estimation and in forming prediction intervals is tested in simulation studies. Two case studies related to solar energy are used to illustrate Quantile Super Learning: in an \iid\ setting, we predict the physical properties of perovskite materials for photovoltaic cells, and in an online setting we forecast ground solar irradiance based on output from dynamic weather ensemble models.
\end{abstract}


\end{frontmatter}

\section{Introduction}

Estimating the quantiles of an outcome conditional on covariates is a foundational task in statistics. Many algorithms have been developed for that purpose, including versions of linear regression, generalized additive models, random forests, and neural networks, to name only a few  \citep{koenker2005quantreg, athey2019grf, cannon2011qrnn, fasiolo2020qgam}. On any particular dataset, however, it is almost never known a-priori which method will perform best. Ensemble algorithms combine the predictions of multiple candidate algorithms according to their empirical performance, obviating the need to choose between the available algorithms in advance. In particular, Super Learning combines the predictions of candidate algorithms according to their cross-validated risk with respect to a loss function chosen on a case by case basis depending on the task at hand \citep{vanderLaan2007superlearner}.

In this work we present the Quantile Super Learner (QSL), a novel ensemble learning algorithm tailored to estimating conditional quantiles. 
To illustrate the main ideas, suppose we have $n$ independent and identically distributed (\iid)\ observations $(O_1, \dots, O_n)$ of a generic variable $O = (X, Y)$ drawn from a law $P_0$, with $X \in \mathcal{X}$ a set of covariates and $Y \in \mathbb{R}$ a univariate outcome. Our goal is to estimate the $\alpha$-conditional quantile of $Y$ given $X$, denoted $\psi^\alpha_{P_0}(X)$ where $\psi^\alpha_{P_0} \in \Psi$ with $\Psi$ the set of all functionals mapping $\mathcal{X} \to \mathbb{R}$. At our disposal are a set of algorithms $\hat{\psi}_k^\alpha$, $k = 1, \dots, K$ that map a dataset to an estimator of the $\alpha$-conditional quantile. No assumptions are necessary about how these algorithms work, and in our simulations and case studies we use a variety methods including quantile regression, gradient boosting machines, and neural networks, among others. 

As we are not likely to know a-priori which algorithm will perform best for a particular dataset, we choose between them by evaluating their performance with respect to the quantile loss function $L^\alpha$, given for any $\psi \in \Psi$ by
\begin{align}
    (x,y) \mapsto L^\alpha(\psi)(x, y) &= \begin{cases}
        \alpha | y - \psi(x) |, & \text{ if } y > \psi(x) \\
        (1 - \alpha) | y - \psi(x) |, & \text{ if } y \leq \psi(x).
    \end{cases}
\end{align}
We choose this loss function because its expected value under $P_0$ is minimized by the true conditional quantile functional:
\begin{align}
    \psi^\alpha_{P_0} \in \argmin_{\psi \in \Psi} \mathbb{E}_{P_0}[L^\alpha(\psi)(X, Y)] = \argmin_{\psi \in \Psi} R_{P_0}^\alpha(\psi),
\end{align}
where we have defined $R_{P_0}^\alpha(\psi)$ to be the \textit{risk} of the algorithm $\psi$ with respect to $P_0$. If we had oracular knowledge of the true data generating distribution $P_0$ we could evaluate the true risk $R_{P_0}^\alpha$ directly for the output of each algorithm and choose the one with the lowest risk. In practice we do not have access to $P_0$, so we must approximate the true risk. The strategy used in Super Learning is to use a cross-validated risk as an estimator of the true risk, and then select the algorithm that minimizes this cross-validated risk. Our main theoretical results are in the form of a bound on the difference between the risk of the algorithm selected using the cross-validated risk and the risk of the algorithm that minimizes the true risk, following the strategy of \citep{wu2022huber}. Notably, these results are established without any regularity assumptions on the data generating distribution.

So far, we have discussed the case where the data are \iid\ draws from a probability law. However, in many scenarios \iid\ assumptions are not justified, such as for time-series data. The key difference between Super Learning in the \iid\ and sequential cases is in the cross-validation scheme used to calculate the empirical risk of the candidate algorithms. In the \iid\ case, $V$-fold cross-validation is typically used, in which the dataset is split uniformly at random into $V$ couples of training and test sets. For the online setting, the collection of all previously observed data is used as a training set and the next observation (or set of observations) is used as a test set. Our main theoretical result for the online setting is similar to that of the \iid\ result, providing bounds on the excess risk of the algorithm chosen based on minimizing the online cross-validated risk. However, for the online setting it is necessary to introduce mild regularity assumptions on the data-generating process. Using ideas from \cite{steinwart2011}, we make a margin assumption concerning the behavior of the conditional law around the quantile being estimated, which allows us to  establish excess risk bounds.

One of the potential use cases of the quantile Super Learner in both \iid\ and online settings is to form prediction intervals by separately estimating a lower and upper quantile. Intuitively, we would expect that if we can do a good job estimating each of these quantiles, then the prediction intervals will also perform well. However, there is no theoretical guarantee that such quantile estimates will yield prediction intervals with desired frequency characteristics in finite samples. We investigate this use case empirically in simulations and case studies.

\paragraph{Prior Work} Early proposals of combining predictions from multiple base learners include stacked generalization, described by \cite{wolpert1992stacking} and \cite{breiman1996stacking}. The theoretical foundations of the approach were formalized by \cite{vanderLaan2003unified}, \cite{vanderLaan2006cvepsilonnet}, and \cite{vandervaart2006oracles}; the name ``Super Learner" was subsequently coined by \cite{vanderLaan2007superlearner}. Early development of Super Learning primarily focused on estimating conditional means using a squared-error loss function in the context of \iid\ data. Since then, extensions to other loss functions have included the Area Under the Curve (AUC) loss function \citep{ledell2016auc} and the Huber loss function \citep{wu2022huber}, among others. Theory for online (sequential) ensemble learning within the Super Learning framework that incorporates the statistical dependence of the data was developed in \cite{benkeser2018onlinesuperlearner} and \cite{ecoto2021onlinesuperlearner}. \cite{fakoor2023qreg} provide a comprehensive account of ensemble methods for quantile estimation, including detailed empirical comparisons. Our work is in a similar vein, but provides formal theoretical guarantees for the proposed Super Learning based approach. \cite{sun2023} propose an online estimator for quantiles within a parametric model and using a smooth approximation of the quantile loss function. In this framework, they show that their ``renewable estimator" is consistent, asymptotically normal, and that it enjoys an oracle property. In contrast, our approach is fully non-parametric.

\paragraph{Outline} The rest of the paper unfolds as follows. In Section \ref{section:iid} we develop the QSL for \iid\ data and establish an oracle inequality showing that the estimator is asymptotically equivalent to the best performing candidate algorithm. In Section \ref{section:sequential} we extend the QSL to the setting of sequential data and establish similar oracle inequalities as in the \iid\ case. In Section \ref{section:simulations} we investigate the finite-sample performance of the QSL in simulations. In Section \ref{section:case-studies} we present two case studies related to solar energy: predicting the physical properties of perovskite materials for photovoltaic applications and forecasting solar irradiance based on the output of deterministic numerical weather prediction models.

\section{Independent setting}
\label{section:iid}
Our initial results develop a QSL in a setting where the observed data represent \iid\ draws from an underlying distribution. Let $(O_1, \dots, O_n)$ be $n$ \iid\ observations of a generic variable $O = (X, Y)$ from the law $P_0$ on $\mathcal{O} = \mathcal{X} \times \mathbb{R}$, where $X \in \mathcal{X}$ is a set of covariates and $Y \in \mathbb{R}$ is a univariate outcome. We assume that $P_0$ falls in statistical model $\mathcal{M}$. 
For all $P \in \mathcal{M}$ and $\alpha \in (0, 1)$, define the $\alpha$-quantile of $Y$ conditional on $X$ as the possibly set-valued functional
\begin{align}
\psi^\alpha_P(X) &:= \{ y \in \mathbb{R} : P(Y < y \mid X) \geq \alpha , P(Y \geq y \mid X) \geq 1 - \alpha \}.
\end{align}
Our goal is to estimate $x \mapsto \psi^\alpha_{P_0}(x)$ for a given $\alpha$, the conditional $\alpha$-quantile under the true data generating distribution $P_0$. For the \iid\ setting, we make two key simplifying assumptions:
\begin{iidassumption}[Unique Quantiles]
\label{assumption:unique-quantiles-iid}
    The quantile function $\psi_{P_0}^\alpha(x)$ is a singleton for $P_X$-almost all $x \in X$.
\end{iidassumption}
\begin{iidassumption}[Outcome Boundedness]
\label{assumption:boundedness}
It holds $P_0$ almost surely that $|Y| \leq C_0 < \infty$. Note that in practice, the value of $C_0$ does not need to be known.
\end{iidassumption}

Next, we establish that $\psi^\alpha_P$ is a minimizer of a particular loss function. Let $\Psi$ be the set of all (measurable) functions mapping $\mathcal{X}$ to $\mathbb{R}$. The quantile loss function $L^\alpha$ for the $\alpha$-quantile is given by
\begin{align}
    (x,y) \mapsto L^\alpha(\psi)(x, y) &:= \begin{cases}
        \alpha | y - \psi(x) |, & \text{ if } y > \psi(x) \\
        (1 - \alpha) | y - \psi(x) |, & \text{ if } y \leq \psi(x)
    \end{cases}
\end{align}
for any $(x, y) \in \mathcal{O}$ and $\psi \in \Psi$.
The \textit{risk} of $\psi$ under the loss function $L^\alpha$ relative to a distribution $P \in \mathcal{M}$ is defined as
\begin{align}
    \label{eq:risk}
    R^\alpha_P(\psi) := \mathbb{E}_P[L^\alpha(\psi)(O)] = PL^\alpha(\psi),
\end{align}
where we use the notation $Pf = \mathbb{E}_P[f(O)] = \int f dP$. Note that   the true conditional quantile function minimizes the quantile loss:
\begin{align}
    \psi^\alpha_{P_0} = \argmin_{\psi \in \Psi} R^\alpha_{P_0}(\psi).
\end{align}
This well-known fact serves as the basis of quantile regression \citep{gneiting2011quantiles, koenker2005quantreg}.

\subsection{Super Learning}
An \textit{algorithm} to learn $\psi^\alpha_{P_0}$ is a function mapping any finite set $\{ o_1, \dots, o_M \}$ of $M$ elements of $\mathcal{O}$, viewed as the measure $M^{-1} \sum_{m=1}^M \text{Dirac}(o_m)$, to an element of~$\Psi$. 
Suppose we have $K$ such algorithms $\widehat{\psi}^{\alpha}_{1}, \dots, \widehat{\psi}^{\alpha}_{K}$ which seek to learn $\psi^\alpha_{P_0}$.  Super Learning amounts to identifying which of the algorithms performs best.

\paragraph{Discrete Super Learner.} 
The discrete Super Learner identifies the best performing algorithm among the candidate algorithms as measured by their cross-validated risks. To formalize the cross-validation scheme, we introduce
$B_{n}\in\{0,1\}^{n}$, a random  vector drawn independently of $O_{1},
\ldots, O_{n}$ such that $\sum_{i=1}^{n} B_{n}(i) \approx np$ for some user-supplied proportion $p$. The observation $O_i$ falls in the training set if $B_n(i) = 0$, and in the testing set if $B_n(i) = 1$.
The $B_n$-specific training and testing datasets are represented by the empirical distributions $P_{n,B_{n}}^{0}$ and $P_{n,B_{n}}^{1}$. For instance, to implement $V$-fold cross-validation, we draw $B_{n}$ from the uniform distribution on $\{b_1, \ldots, b_{V}\} \subset \{0, 1\}^n$ where each $b_v$ satisfies $\sum_{i=1}^n b_v(i) \approx n/V$ (a proportion $(V - 1) / V$ of data are used for training, and the rest for testing) and, for every $1 \leq i \leq n$, $\sum_{v=1}^V b_v(i) = 1$ (each observation is used once for testing). 

The oracle cross-validated risk of algorithm $\widehat{\psi}^{\alpha}_{k}$ with respect to a distribution $P \in \mathcal{M}$ is then defined by
\begin{align}
    \label{eq:oracle-risk}
    \widetilde{R}^\alpha_{n,P}(\widehat{\psi}^{\alpha}_{k}) := \mathbb{E}_{B_n}\left[ P L^\alpha(\widehat{\psi}^{\alpha}_{k}(P_{n,B_n}^0))\right].
\end{align}
The empirical version of the oracle cross-validated risk is simply obtained by substituting $P_{n,B_n}^1$ for $P$ in \eqref{eq:oracle-risk}:
\begin{align}
    \widehat{R}^\alpha_n(\widehat{\psi}^{\alpha}_ k) := \mathbb{E}_{B_n}\left[ P_{n,B_n}^1 L^\alpha(\widehat{\psi}^{\alpha}_{k}(P_{n,B_n}^0))\right].
\end{align}
The discrete Super Learner selector minimizes the empirical cross-validated risk:
\begin{align}
    \widehat{\kappa}_n := \argmin_{k\in \llbracket K \rrbracket } \widehat{R}_n^\alpha(\hat{\psi}^{\alpha}_{k}),
\end{align}
using the notation $\llbracket K \rrbracket := \{1, 2, \dots, K \}$.
The corresponding algorithm $\widehat{\psi}^{\alpha}_{\widehat{\kappa}_n}$ is referred to as the discrete Super Learner.

\paragraph{Continuous Super Learner}
The continuous Super Learner considers a richer class of algorithms taking the form of convex combinations of the original candidate algorithms. Let $\Pi$ be the $K$-simplex: that is, the set of $\pi \in (\mathbb{R}_+)^K$ such that  $\sum_{k=1}^K \pi_k = 1$. The new generic candidate algorithms take the form
\begin{align}
    \label{convex-combination}
    \widehat{\psi}^{\alpha}_{\pi} := \sum_{k=1}^K \pi_k \widehat{\psi}^{\alpha}_{k}
\end{align}
for any $\pi \in \Pi$.
Let $\Pi_n$ be a finite subset of $\Pi$ such that the cardinality of $\Pi_n$ grows at most polynomially with $n$. The continuous Super Learner is found by finding the weights $\pi \in \Pi_n$ that minimize the empirical cross-validated risk:
\begin{align}
    \widehat{\pi}_n \in \argmin_{\pi \in \Pi_n} \widehat{R}^\alpha_n(\widehat{\psi}^{\alpha}_{\pi}).
\end{align}
The algorithm $\widehat{\psi}^{\alpha}_{\widehat{\pi}_n}$ is referred to as the continuous Super Learner.
Note that the continuous Super Learner is simply the discrete Super Learner when the collection of candidate algorithms is $\{ \widehat{\psi}_\pi^\alpha : \pi \in \Pi_n \}$. As such, we focus on analyzing the properties of the discrete Super Learner, as the results carry over to the continuous Super Learner.

\subsection{Oracle Inequalities}
We compare the discrete Super Learner against an oracle selector which identifies the candidate algorithm that has the best oracle cross-validated risk with respect to the law $P_0$:
\begin{align}
    \widetilde{\kappa}_n := \argmin_{k \in \llbracket K \rrbracket} \widetilde{R}^{\alpha}_{n,P_0}(\widehat{\psi}^{\alpha}_{k}).
\end{align}
The corresponding algorithm $\widehat{\psi}^\alpha_{\widetilde{\kappa}_n}$ is referred to as the oracle Super Learner.

Let $\widetilde{\psi}^\alpha_{P_0}$ denote the oracle algorithm that constantly outputs the true conditional $\alpha$-quantile $\psi^\alpha_{P_0}$. 
Our theoretical results compare the excess risk of the discrete Super Learner to the excess risk of the oracle Super Learner, that is
\begin{align}
    \widetilde{R}_{n,P_0}^\alpha(\widehat{\psi}^\alpha_{\widehat{\kappa}_n}) - \widetilde{R}_{n,P_0}^\alpha(\widetilde{\psi}^\alpha_{P_0}) \text{\,  vs.  \,} \widetilde{R}_{n,P_0}^\alpha(\widehat{\psi}^\alpha_{\widetilde{\kappa}_n}) - \widetilde{R}_{n,P_0}^\alpha(\widetilde{\psi}^\alpha_{P_0}),
\end{align}
where 
\begin{align}
    \widehat{R}^\alpha_{n,P_0}(\widetilde{\psi}^\alpha_{P_0}) &:= \mathbb{E}_{B_n}\left[ P_{n,B_n}^1 L^\alpha(\psi^\alpha_{P_0}) \right].
\end{align}

We are now ready to state the main result of this section, an excess risk bound for the discrete Super Learner.
\begin{theorem}[Excess risk bounds for discrete Super Learner]
    \label{theorem:iid-oracle-inequality}
    Assume that the number of candidate algorithms grows at most polynomially in $n$: that is, $K = O(n^a)$ for some $a > 0$. Also assume that, for all $k \in \llbracket K \rrbracket$,  $\hat{\psi}^\alpha_k$ only outputs functions which are uniformly bounded by $C_0$. Then
    \begin{align}
        \mathbb{E}_{P_0}\left[\widetilde{R}_{n,P_0}^\alpha(\widehat{\psi}^\alpha_{\widehat{\kappa}_n}) - \widetilde{R}_{n,P_0}^\alpha(\widetilde{\psi}^\alpha_{P_0})\right] \leq \mathbb{E}_{P_0}\left[\widetilde{R}_{n,P_0}^\alpha(\widehat{\psi}^\alpha_{\widetilde{\kappa}_n}) - \widetilde{R}_{n,P_0}^\alpha(\widetilde{\psi}^\alpha_{P_0})\right] + O\left( \frac{\log(n)}{n^{1/2}} \right).
    \end{align}
\end{theorem}

\section{Sequential setting}
\label{section:sequential}
In the sequential setting we gain access to the observations in batches.  In particular, we are interested in the setting where, at each time point, we gain access to a new batch of observations, one for each element in an index set $\mathcal{J}$. For example, we may have multiple time series corresponding to several locations, where each location generates a new observation at each timepoint. Note that if $|\mathcal{J}| = 1$, then the problem reduces to the case where we observe a single time series.

Let $(\bar{O}_{t})_{t\geq 1}$ be a time-ordered sequence where $t$ indexes time. Each $\bar{O}_t$ is the set of observations $\bar{O}_t = ( O_{j, t} : j \in \mathcal{J})$. The observations then decompose as $O_{j,t} = (X_{j,t}, Y_{j,t})$ for $j \in \mathcal{J}, t \geq 1$. Let $P_0$ be the joint law of the observed data, which we assume falls in a statistical model $\mathcal{M}$. For all $t \geq 2$, introduce the $\sigma$-field $F_{t-1} := \sigma \left( O_{j, t^\prime} : j \in \mathcal{J}, 1 \leq \tau < t \right)$ generated by past observations (with $F_0 := \emptyset$ by convention). 
For all $P \in \mathcal{M}$ and $(j,\tau) \in \mathcal{J} \times \mathbb{N}^*$ define the $\alpha$-quantile of $Y$ conditional on $X$ at location $j$ and time $\tau$ as the possibly set-valued 
\begin{align}
\psi^\alpha_{P, j, t}(X_{j,t}) &:= \{ y \in \mathbb{R} : P(Y_{j,t} < y \mid X_{j,t}) \geq \alpha, P(Y_{j,t} \geq y \mid X_{j,t}) \geq 1 - \alpha \},
\end{align}
where $\alpha \in (0, 1)$ is taken as fixed. The sequential setting requires several more assumptions in addition to those adopted in the \iid\ case. First, we make the simplifying assumption of common support of the covariates across all locations and time points.
\begin{assumption}[Common Support]
    \label{assumption:common-support}
    For all $P \in \mathcal{M}$ there exists $\mathcal{S}$ such that, for all $(j, \tau) \in \mathcal{J} \times \mathbb{N}^*$,
    \begin{align}
        \mathrm{Supp}\left( P_{X_{j,\tau}} \right) = \mathcal{S},
    \end{align}
    where $P_{X_{j,\tau}}$ is the marginal law of $X$ under $P$ at location $j$ and time $\tau$.
\end{assumption}
As in the \iid\ setting, we will also assume the existence of unique quantiles and that the outcomes are uniformly bounded.
\begin{assumption}[Unique Quantiles]
\label{assumption:unique-quantiles-online}
    For all $(j,\tau) \in \mathcal{J} \times \mathbb{N}^*$, the $\alpha$-quantile  $\psi_{P_0}^\alpha(x)$ is a singleton for all $x \in \mathcal{S}$.
\end{assumption}
\begin{assumption}[Outcome Boundedness]
\label{assumption:online-boundedness}
It holds $P_0$-almost surely that, for all $(j,\tau) \in \mathcal{J} \times \mathbb{N}^*$, $|Y_{j,t}| \leq C_0 < \infty$.
\end{assumption}
We also make a Markovian assumption that all information about the outcome at a particular location and time point is encoded in prior observations and the covariates for that location and time point.
\begin{assumption}[Markov]
    \label{assumption:online-markov}
    For every $(j, \tau) \in \mathcal{J} \times \mathbb{N}^*$, it holds $P_0$-almost surely that
    \begin{align}
        P_0\left( Y_{j,\tau} | (X_{j,\tau} : j\in \mathcal{J}) \right) = P_0\left( Y_{j,\tau} | X_{j,\tau} \right).
    \end{align}
\end{assumption}
Finally, the next assumption guarantees that we can learn the quantile function based on the observed time series $\bar{O}_1, \dots, \bar{O}_t$.
\begin{assumption}[Stationarity]
    \label{assumption:stationarity}
    There exists $\Psi^\alpha_{P_0} \in \Psi$ such that, for every $(j,\tau) \in \mathcal{J} \times \mathbb{N}^*$, 
    \begin{align}
        \argmin_{\psi \in \Psi} \mathbb{E}_{P_0}[L^\alpha(\psi)(O_{j,\tau})] = \psi^\alpha_{P_0}.
    \end{align}
    (Informally, the assumption states that $\Psi^\alpha_{P_{0},j,\tau} = \Psi^\alpha_{P_0}$ for all $(j,\tau) \in \mathcal{J} \times \mathbb{N}^*$.)
\end{assumption}

We prove the oracle inequalities for Super Learning in the sequential setting by making a regularity assumption on the $P_0$-conditional laws given $F_t$. In order to state the assumption, we need the following two definitions drawn from \citep{steinwart2011}.
\begin{definition}[Quantiles of type $q$ \citep{steinwart2011}]
    \label{def:quantile-type-q}
    Let $Q$ be a distribution with $\mathrm{Supp}(Q) \subset[-1, 1]$. Set arbitrarily $\alpha \in (0, 1)$, let $\psi_Q^\alpha := \{ t \in \mathbb{R} : Q((-\infty, t]) \geq \alpha, Q([t, \infty)) \geq 1 - \alpha \}$ be the $\alpha$-quantile of $Q$, and assume that $\psi_Q^\alpha$ is a singleton. 
    The distribution $Q$ is said to have an $\alpha$-quantile of type $q \in (1, \infty)$ if there exist constants $\alpha_Q \in (0, 2]$ and $b_Q > 0$ such that
    \begin{align}
        Q((\psi_Q^\alpha - s, \psi_Q^\alpha)) &\geq b_Q s^{q-1}, \\
        Q((\psi_Q^\alpha, \psi_Q^\alpha + s)) &\geq b_Q s^{q-1}
    \end{align}
    for all $s \in [0, \alpha_Q]$. We also define $\gamma_Q = b_Q \alpha_Q^{q-1}$.
\end{definition}
Definition \ref{def:quantile-type-q} applies to a distribution with support on a subset of $\mathbb{R}$. The next definition is an extension to distributions defined on a subset of $\mathcal{X} \times \mathbb{R}$. 

\begin{definition}[Quantiles of $p$-average type $q$ \citep{steinwart2011}]
    \label{def:quantile-p-average-type-q}
    Let $p \in (0, \infty]$, $q \in [1, \infty)$, and $Q$ be a distribution on $\mathcal{X} \times \mathbb{R}$ with marginal distribution $Q_X$ of $\mathcal{X}$. Assume that $\mathrm{Supp}(Q(\cdot \mid X = x)) \subset [-1, 1]$ for $Q_X$-almost all $x \in X$. Then $Q$ is said to have an $\alpha$-quantile of $p$-average type $q$ if $Q(\cdot \mid X = x)$ has an $\alpha$-quantile of type $q$ for $Q_X$-almost all $x \in \mathcal{X}$, and if the function $\gamma : \mathcal{X} \to [0, \infty]$ given for $Q_X$-almost all $x \in \mathcal{X}$ by
    \begin{align}
        \gamma(x) := \gamma_{Q(\cdot \mid X = x)} := b_{Q(\cdot \mid X = x)} \alpha_{Q(\cdot \mid X = x)}^{q-1}
    \end{align}
    (as defined in Definition \ref{def:quantile-type-q}) is such that $\gamma^{-1}$ admits a finite moment $\|\gamma^{-1}\|_{p,Q_X}$ of order $p$ under $Q_X$.
\end{definition}

Now we are ready to state the final assumption we use for the oracle inequality in the sequential setting.

\begin{assumption}[Regularity]
    \label{assumption:online-regularity}
    For every $(j,\tau) \in \mathcal{J} \times \mathbb{N}^*$, the $P_0$-conditional law of $O_{j,\tau}$ given $F_{\tau-1}$ has an $\alpha$-quantile of $p$-average type $q$. Moreover, the collection of $\| \gamma^{-1} \|_{p,Q_X}$, where $Q$ ranges over the $P_0$-conditional laws of $O_{j,\tau}$ given $F_{\tau - 1}$, is uniformly bounded by a constant $\Gamma > 0$.
\end{assumption}

\subsection{Super Learner}
\paragraph{Discrete Super Learner}
Suppose as in the \iid\ setting that we have $K$ algorithms $\widehat{\psi}^{\alpha}_1, \dots, \widehat{\psi}^\alpha_K$ to learn $\psi^\alpha_{P_0}$. Each $\widehat{\psi}^\alpha_k$ is a function mapping any finite sequence $\bar{o}_1, \dots, \bar{o}_t$ to an element of $\Psi$. Let $P_{t} := \sum_{\tau=1}^t \mathrm{Dirac}(\bar{o}_\tau)$ be the empirical distribution of the data up to time $t$. Define the $L^\alpha$-loss of $\psi \in \Psi$ with respect to a batch of observations $\bar{o}_t$ as:
\begin{align}
    \bar{L}^\alpha(\psi)(\bar{o}_t) := \frac{1}{|\mathcal{J}|} \sum_{j \in \mathcal{J}} L^\alpha(\psi)(o_{j,t}).
\end{align}
The oracle risk of an algorithm $\widehat{\psi}^\alpha_k$ up to time $t \geq 1$ with respect to a distribution $P \in \mathcal{M}$ is defined as:
\begin{align}
    \widetilde{R}^{\alpha}_{t,P}(\widehat{\psi}^\alpha_k) := \frac{1}{t} \sum_{\tau = 1}^t \mathbb{E}_P \left[ \bar{L}^\alpha(\widehat{\psi}^\alpha_{k}(P_{\tau - 1}))(\bar{O}_\tau) \middle| F_{\tau - 1} \right],
\end{align}
where by convention $\hat{\psi}_k^\alpha(P_{t-1})$ produces an arbitrary constant function (e.g. always zero) for the case $t = 1$.
The empirical risk up to time $t$ of $\hat{\psi}_k^\alpha$ is defined as the empirical counterpart of its oracle risk:
\begin{align}
    \label{eq:online-empirical-risk}
    \widehat{R}^{\alpha}_{t}(\widehat{\psi}_k^\alpha) &:= 
    \frac{1}{t} \sum_{\tau=1}^t \bar{L}^\alpha(\widehat{\psi}^\alpha_k(P_{\tau - 1}))(\bar{O}_\tau) \\
     &= \frac{1}{t|\mathcal{J}|} \sum_{\tau = 1}^t \sum_{j \in \mathcal{J}} L^\alpha(\widehat{\psi}^\alpha_k(P_{\tau - 1}))(O_{j,\tau}).
    \end{align}
At each $t \geq 1$, the discrete online Super Learner selector is formed by finding the algorithm that minimizes the empirical risk up to time $t$:
\begin{align}
    \widehat{\kappa}_t = \argmin_{k \in \llbracket K \rrbracket} \widehat{R}^{\alpha}_t(\widehat{\psi}^\alpha_k).
\end{align}
The algorithm $\hat{\psi}^\alpha_{\hat{\kappa}_t}$ is referred to as the online discrete Super Learner.

\paragraph{Continuous Super Learner}
Convex combinations of the candidate algorithms are formed as in the \iid\ setting \eqref{convex-combination}, yielding new candidate algorithms $\widehat{\psi}^{\alpha}_{\pi}$ (for any $\pi \in \Pi_n)$.  
The continuous online Super Learner selector is then formed by finding weights that minimize the empirical risk in hindsight:
\begin{align}
    \widehat{\pi}_t \in \argmin_{\pi \in \Pi_n} \widehat{R}^\alpha_t(\widehat{\psi}^{\alpha}_{\pi}).
\end{align}
The algorithm $\hat{\psi}^\alpha_{\hat{\pi}_t}$ is referred to as the online continuous Super Learner.

\subsection{Oracle Inequalities}
In this section we establish oracle inequalities for the discrete online Super Learner, based on results for online Super Learner established in \citep{ecoto2021onlinesuperlearner}. The main result is presented below, with the proof to be found in the appendix.

\begin{theorem}
    \label{theorem:online-main-result}
    Assume that the number of candidate algorithms grows at most polynomially in $t$, that is, $K = O(t^a)$ for some $a > 0$. Also assume that, for all $k \in \llbracket K \rrbracket$,  $\hat{\psi}^\alpha_k$ only outputs functions which are uniformly bounded by $C_0$. Then, for $t$ large enough, 
    \begin{align}
        \mathbb{E}_{P_0}\left[ \widetilde{R}^\alpha_{t,P_0}(\widehat{\psi}^\alpha_{\hat{\kappa}_t}) - \widetilde{R}_{t,P_0}^\alpha(\psi_{P_0}^\alpha) \right] \leq 
        \mathbb{E}_{P_0}\left[\widetilde{R}^\alpha_{t, P_0}(\widehat{\psi}^\alpha_{\widetilde{\kappa}_t}) - \widetilde{R}_{t,P_0}^\alpha(\widehat{\psi}_{P_0}^\alpha) \right]
        + O\left( \frac{\log(\log(t))}{t^{1/2}} \right).
    \end{align}
\end{theorem}
Note that the number of locations $|\mathcal{J}|$ is hidden in the ``large enough" $t$ (and not in the order term). Specifically, the larger is the number of locations $|\mathcal{J}|$, the smaller $t$ needs to be for the inequality to hold.

\begin{table}
    \centering
    \small
    \begin{tabular}{|l|l|p{4cm}|}
        \hline
         Algorithm & R package & Citations \\
         \hline
         Distributional Random Forest (DRF) & \texttt{drf} & \cite{michel2021drf, drf2022} \\
         Gradient Boosting Machine (GBM) & \texttt{lightgbm} & \cite{shi2023lightgbm} \\
         Quantile Generalized Additive Models (QGAM) & \texttt{qgam} & \cite{fasiolo2020qgam,fasiolo2021qgampackage} \\
         Quantile Random Forest (QRF) & \texttt{grf} & \cite{athey2019grf} \\
         Quantile Regression (QReg) & \texttt{quantreg} & \cite{koenker2005quantreg,koenker2017handbook} \\
         Quantile Regression Neural Network (QRNN) & \texttt{qrnn} & \cite{cannon2011qrnn,cannon2018qrnn} \\
         \hline
    \end{tabular}
    \caption{Example library of candidate algorithms for \iid\ conditional quantile estimation.}
    \label{tab:base_learners}
\end{table}

\section{Simulation Studies}
\label{section:simulations}
In this section we investigate the finite sample performance of the QSL in the \iid\ and sequential settings. A natural application of quantile estimation is in forming prediction intervals from estimates of a lower and upper quantile, which we investigate in both settings. We used the \texttt{sl3} \texttt{R} package \citep{coyle2021sl3-rpkg} to implement the \iid\ QSL algorithm. For the sequential setting, we compare the online quantile Super Learner to two algorithms from the online aggregation of experts literature: Exponentially Weighted Average (EWA; \cite{cesalugosi2006onlinelearning}) and Bernstein Online Aggregation (BOA; \cite{wintenberger2017boa}). Both algorithms are implemented in the \texttt{opera} R package \citep{opera2023r}, which uses an adaptive procedure to fine tune their learning rates. We added the QSL as an additional method to the \texttt{opera} package to facilitate its use and comparison with other methods. Code for the simulations and case studies can be found at \url{https://github.com/herbps10/QuantileSuperLearner}.

\subsection{Independent setting}
Simulated datasets consisted of $N_1 \in \{ 250, 500, 1000 \}$ \iid\ draws $(X_i, Y_i)$, $i = 1, \dots, N_1$, of a generic variable $(X, Y)$. The covariates forming $X = (X_1, X_2, X_3, X_4, X_5)$ were drawn independently from the uniform distribution on $[0, 1]$. Conditional on $X$, the outcome $Y$ was chosen to be a linear combination of smooth and non-smooth functions of the covariates:
\begin{align}
    Y = \sin(2X_1) + |X_2| - 0.5 X_1 X_3 + \lfloor X_4\rfloor + \epsilon,
\end{align}
where $\epsilon \sim N(0, 0.1)$. An additional $N_2 = 1000$ observations $(X_i, Y_i)$, $i = N_1 + 1, \dots, N_1 + N_2$, were drawn from the same data generating process and used as a validation set to evaluate the performance of the Super Learner in predictions of unseen data. Overall, $50$ learning and testing datasets were generated for each sample size $N_1$.

We estimated the $\alpha$-quantile with $\alpha \in \{ 0.025, 0.05, 0.1, 0.5, 0.9, 0.95, 0.975 \}$ for each simulated dataset using separate QSLs. The candidate algorithms included gradient boosting machines, quantile regression, quantile neural networks, quantile random forests, and quantile generalized additive models (see Table~\ref{tab:base_learners} for references). We compared the QSL against the candidate algorithms by calculating their empirical risks with respect to the testing dataset. To do so, we defined the empirical risk of an algorithm $\widehat{\psi}^\alpha$ as
\begin{align}
    \label{eq:emp-risk-iid}
    \mathrm{EmpRisk}(\alpha) := \frac{1}{N_2} \sum_{i=N_1 + 1}^{N_1 + N_2} L^\alpha(\widehat{\psi}^\alpha(P_{N_1}))(X_i, Y_i).
\end{align}
In addition, 80\%, 90\%, and 95\% prediction intervals were formed using the 10\%, 5\%, and 2.5\% and 90\%, 95\%, and 97.5\% quantile estimates as the lower and upper interval bounds, respectively. The empirical coverage of a $(1-\beta) \times 100\%$ prediction interval built using algorithms $\widehat{\psi}^{\beta/2}$ and $\widehat{\psi}^{1 - \beta/2}$ was defined as
 \begin{align}
    \label{eq:emp-cov-iid}
    \mathrm{EmpCov}(\beta) := \frac{1}{N_2} \sum_{i=N_1 + 1}^{N_1 + N_2} \mathbb{I}\left[\widehat{\psi}^{\beta/2}(P_{N_1})(X_i) \leq Y_i \leq \widehat{\psi}^{1 - \beta/2}(P_{N_1})(X_i)\right].
 \end{align}

\paragraph{Results} The empirical risk 
\eqref{eq:emp-risk-iid} results are shown in Table~\ref{tab:iid-simulation-results}. QSL achieved the best (or tied for best) quantile risk for all quantiles and sample sizes. However, the empirical coverage prediction intervals formed using QSL estimates, as shown in Table~\ref{tab:iid-simulation-results-coverage}, did not necessarily perform as well compared to the candidate algorithms.

\begin{table}
    \centering
    \small
    \begin{tabular}{|r|l|rrrrrrr|}
    \hline
    & & \multicolumn{7}{c|}{$\mathrm{EmpRisk}(\alpha)$ } \\
    $N_1$ & Algorithm & $\alpha = 0.025$ & $\alpha = 0.05$ & $\alpha = 0.1$ & $\alpha = 0.5$ & $\alpha = 0.9$ & $\alpha = 0.95$ & $\alpha = 0.975$ \\
    \hline
    250 & GRF & 0.063 & 0.11 & 0.17 & 0.33 & 0.17 & 0.11 & 0.069\\
     & GBM & 0.06 & 0.084 & 0.11 & \textbf{0.16} & 0.1 & 0.079 & 0.058\\
     & QGAM & 0.034 & 0.057 & 0.098 & 0.29 & 0.1 & 0.058 & \textbf{0.035}\\
     & QRNN & 0.052 & 0.083 & 0.13 & 0.32 & 0.13 & 0.083 & 0.052\\
     & QReg & 0.057 & 0.1 & 0.17 & 0.41 & 0.18 & 0.1 & 0.058\\
     & QSL & \textbf{0.033} & \textbf{0.053} & \textbf{0.088} & \textbf{0.16} & \textbf{0.09} & \textbf{0.057} & \textbf{0.035}\\
    \hline
    500 & GRF & 0.058 & 0.095 & 0.15 & 0.27 & 0.16 & 0.098 & 0.061\\
     & GBM & 0.05 & 0.065 & 0.083 & \textbf{0.12} & 0.084 & 0.067 & 0.051\\
     & QGAM & 0.032 & 0.053 & 0.092 & 0.27 & 0.092 & 0.053 & 0.032\\
     & QRNN & 0.041 & 0.072 & 0.12 & 0.32 & 0.12 & 0.073 & 0.042\\
     & QReg & 0.056 & 0.1 & 0.17 & 0.4 & 0.17 & 0.1 & 0.056\\
     & QSL & \textbf{0.028} & \textbf{0.045} & \textbf{0.07} & \textbf{0.12} & \textbf{0.074} & \textbf{0.049} & \textbf{0.031}\\
    \hline
    1000 & GRF & 0.051 & 0.084 & 0.13 & 0.2 & 0.14 & 0.088 & 0.054\\
     & GBM & 0.038 & 0.051 & 0.064 & \textbf{0.092} & 0.063 & 0.051 & 0.04\\
     & QGAM & 0.029 & 0.05 & 0.088 & 0.27 & 0.088 & 0.05 & 0.029\\
     & QRNN & 0.039 & 0.068 & 0.12 & 0.31 & 0.12 & 0.07 & 0.038\\
     & QReg & 0.055 & 0.099 & 0.17 & 0.4 & 0.17 & 0.099 & 0.054\\
     & QSL & \textbf{0.024} & \textbf{0.038} & \textbf{0.057} & \textbf{0.092} & \textbf{0.058} & \textbf{0.041} & \textbf{0.027}\\
    \hline
    \end{tabular}
    \caption{\label{tab:iid-simulation-results} Results for the \iid\ simulation study. The QSL and candidate algorithms were trained on a learning dataset of $N_1$ observations and evaluated \eqref{eq:emp-risk-iid}
 on a testing dataset of $N_2 = 1000$ observations.}
\end{table}

\begin{table}
    \centering
    \begin{tabular}{|r|l|rrr|}
    \hline
    & & \multicolumn{3}{c|}{$\mathrm{EmpCov}(\beta)$} \\
    $N_1$ & Algorithm & $(1 - \beta) = 0.8$ & $(1 - \beta) = 0.9$ & $(1 - \beta) = 0.95$ \\
    \hline
    250 & GRF & 88.2\% & 94.9\% & 97.2\%\\
     & GBM & 56.1\% & 70.7\% & 81.7\%\\
     & QGAM & 87.9\% & 96.9\% & 99.3\%\\
     & QRNN & 75.1\% & 83.7\% & 87.2\%\\
     & QReg & \textbf{78.8\%} & \textbf{88.4\%} & \textbf{93.2\%}\\
     & QSL & 82.7\% & 94.8\% & 98.3\%\\
    \hline
    500 & GRF & 89.3\% & 95.3\% & 97.6\%\\
     & GBM & 53\% & 67.1\% & 79.3\%\\
     & QGAM & 89.8\% & 97.8\% & 99.7\%\\
     & QRNN & 77.6\% & 86.8\% & 91.5\%\\
     & QReg & \textbf{79.1\%} & \textbf{89.1\%} & \textbf{93.8\%}\\
     & QSL & 77.2\% & 92.8\% & 98.1\%\\
    \hline
    1000 & GRF & 92\% & 96.7\% & 98.5\%\\
     & GBM & 53.1\% & 66.5\% & 78.6\%\\
     & QGAM & 89.2\% & 97.6\% & 99.7\%\\
     & QRNN & 79.1\% & 88.6\% & 93.6\%\\
     & QReg & \textbf{79.9\%} & 89.5\% & \textbf{94.4\%}\\
     & QSL & 73\% & \textbf{89.7\%} & 97.4\%\\
    \hline
    \end{tabular}
    \caption{\label{tab:iid-simulation-results-coverage} Results for the \iid\ simulation study. The QSL and candidate algorithms were trained on a learning dataset of $N_1$ observations and evaluated \eqref{eq:emp-cov-iid}
 on a testing dataset of $N_2 = 1000$ observations.}
\end{table}

\subsection{Sequential setting}
For the sequential setting we augmented the data generating process from the \iid\ simulation study to induce temporal dependence. We simulated $Y_t$, $t = 1, \dots, T = 2000$ following
\begin{align}
    Y_t = \sin(2X_1) + |X_2| - 0.5X_1 X_3 + \lfloor X_4 \rfloor + \epsilon_t,
\end{align}
where $(\epsilon_t)_{t \leq T}$ is now drawn from an AR(1) process:
\begin{align}
    \epsilon_1 &\sim N(0, \sigma^2 / (1 - \rho^2)) \quad \text{ and, for $1 < t \leq T$, } \quad \epsilon_t  \sim N(\rho \cdot \epsilon_{t-1}, \sigma^2),
\end{align}
with $\rho \in (0, 1)$ an autoregressive parameter and $\sigma > 0$ a scale parameter. 

We estimated the same set of $\alpha$-quantiles as before for each simulated dataset using separate online QSLs. The candidate algorithms included gradient boosting machines, quantile regression, quantile neural networks, quantile random forests, and quantile generalized additive models (see Table~\ref{tab:base_learners}). To save computational time, each of the candidate algorithms was fit once using a single training set of all observations from $t = 1$ to $t = T / 2 = 1000$. This amounts to substituting $P_{T/2}$ for $P_{\tau-1}$ in 
\eqref{eq:online-empirical-risk}.

The empirical risk and coverage of the candidate algorithms are defined as in \eqref{eq:emp-risk-iid} and \eqref{eq:emp-cov-iid}, substituting $T/2$ for $N_1$ and $N_2$. The empirical risk and coverage of the online algorithms (QSL, EWA, and BOA) is defined differently as they are updated for each $t$. For an online algorithm $\widehat{\psi}^\alpha$ the final empirical risk is defined as
\begin{align}
    \label{eq:emp-risk-online}
    \mathrm{EmpRisk}(\alpha) := \frac{1}{T/2} \sum_{t=T/2 + 1}^{T} L^\alpha(\widehat{\psi}_t^\alpha)(X_{t}, Y_{t}),
\end{align}
where $\widehat{\psi}^\alpha_t$ is the output of the algorithm using all data before time $t$.
Similarly, the final empirical coverage for $(1 - \beta) \times 100\%$ prediction intervals formed from algorithms $\widehat{\psi}^{\beta / 2}$ and $\widehat{\psi}^{1 - \beta/2}$ is defined as
\begin{align}
    \label{eq:emp-cov-online}
    \mathrm{EmpCov}(\beta) := \frac{1}{T/2} \sum_{t=T/2 + 1}^{T} \mathbb{I}\left[ \widehat{\psi}_t^{\beta / 2}(X_{t}, Y_{t}) \leq Y_{t} \leq \widehat{\psi}_t^{1 - \beta / 2}(X_{t}, Y_{t}) \right],
\end{align}
where $\widehat{\psi}_t^{\beta / 2}$ and $\widehat{\psi}_{t}^{1 - \beta / 2}$ are the output of the algorithms $\widehat{\psi}^{\beta / 2}$ and $\widehat{\psi}^{1 - \beta / 2}$ using data before time $t$, respectively.

\paragraph{Results} The online empirical risk 
\eqref{eq:emp-risk-online} results are shown in Table~\ref{tab:online-simulation-results-risk}. The algorithms yielded similar empirical risks, with QSL having slightly lower risks for most quantiles and settings of $\rho$. The prediction intervals formed using the quantile estimates from each algorithm tended to undercover (see Table~\ref{tab:online-simulation-results-coverage}), especially for QSL. 

\begin{table}
    \centering
    \small
    \begin{tabular}{|l|l|rrrrrrr|}
    \hline
    & & \multicolumn{7}{c|}{$\mathrm{EmpRisk}(\alpha)$ } \\
    $\rho$& Algorithm & $\alpha = 0.025$ & $\alpha = 0.05$ & $\alpha = 0.1$ & $\alpha = 0.5$ & $\alpha = 0.9$ & $\alpha = 0.95$ & $\alpha = 0.975$\\
    \hline
    0 & BOA & 0.024 & 0.038 & 0.056 & 0.3 & 0.055 & 0.04 & 0.028\\
     & EWA & 0.024 & 0.038 & 0.057 & 0.29 & 0.056 & 0.04 & 0.027\\
     & QSL & \textbf{0.023} & \textbf{0.036} & \textbf{0.055} & \textbf{0.27} & \textbf{0.054} & \textbf{0.038} & \textbf{0.026}\\
    \hline
    0.5 & BOA & 0.024 & 0.038 & 0.058 & 0.3 & 0.058 & 0.041 & 0.028\\
     & EWA & 0.024 & 0.038 & 0.059 & 0.3 & 0.059 & 0.041 & 0.028\\
     & QSL & \textbf{0.023} & \textbf{0.036} & \textbf{0.056} & \textbf{0.28} & \textbf{0.056} & \textbf{0.039} & \textbf{0.026}\\
    \hline
    0.9 & BOA & \textbf{0.03} & 0.047 & 0.073 & 0.38 & 0.072 & 0.048 & 0.031\\
     & EWA & \textbf{0.03} & 0.046 & \textbf{0.072} & 0.36 & 0.072 & 0.048 & 0.03\\
     & QSL & \textbf{0.03} & \textbf{0.045} & \textbf{0.072} & \textbf{0.34} & \textbf{0.071} & \textbf{0.047} & \textbf{0.029}\\
    \hline
    0.99 & BOA & 0.068 & 0.084 & 0.14 & 0.61 & 0.14 & 0.082 & 0.048\\
     & EWA & \textbf{0.067} & \textbf{0.078} & \textbf{0.13} & \textbf{0.56} & \textbf{0.13} & \textbf{0.078} & \textbf{0.047}\\
     & QSL & 0.069 & 0.086 & 0.14 & 0.58 & 0.14 & 0.083 & 0.049\\
    \hline
    \end{tabular}
    \caption{\label{tab:online-simulation-results-risk} Results for the online simulation study in terms of empirical risk \eqref{eq:emp-risk-online}.}
\end{table}

\begin{table}
    \centering
    \begin{tabular}{|l|l|rrr|}
    \hline
    & & \multicolumn{3}{c|}{$\mathrm{EmpCov}(\beta)$} \\
    $\rho $& Algorithm & $(1 - \beta) = 0.8$ & $(1 - \beta) = 0.9$ & $(1 - \beta) = 0.95$ \\
    \hline
    0 & BOA & \textbf{71.8\%} & \textbf{88.2\%} & \textbf{93.5\%}\\
     & EWA & 71.7\% & 86.9\% & 92.1\%\\
     & QSL & 70.9\% & 86.4\% & 91.7\%\\
    \hline
    0.5 & BOA & \textbf{72.1\%} & \textbf{88.4\%} & \textbf{93.6\%}\\
     & EWA & 71.9\% & 87\% & 92.1\%\\
     & QSL & 70.9\% & 86.5\% & 92\%\\
    \hline
    0.9 & BOA & \textbf{73.1\%} & \textbf{88.7\%} & \textbf{92.9\%}\\
     & EWA & 72.8\% & 87.1\% & 91.9\%\\
     & QSL & 71.7\% & 86.5\% & 91.5\%\\
    \hline
    0.99 & BOA & 73.3\% & 86.3\% & 87.3\%\\
     & EWA & \textbf{73.9\%} & \textbf{86.5\%} & \textbf{87.4\%}\\
     & QSL & 71.2\% & 83.8\% & 86\%\\
    \hline
    \end{tabular}
    \caption{\label{tab:online-simulation-results-coverage} Results for the online simulation study in terms of empirical coverage \eqref{eq:emp-cov-online}.}
\end{table}

\section{Case Studies}
\label{section:case-studies}
In this section we present two case studies based on solar energy applications, one for the \iid\ setting and one for the online setting.

\subsection{Perovskite energy formation and bandgap prediction}
\label{section:case-study-perovskite}
A critical component of photovoltaic (PV) cells is the material used for the light-absorbing semiconductor layer. The use of materials with perovskite crystal structures has been the subject of significant recent research, leading to the energy efficiency of perovskite-based PV cells increasing rapidly from 3.8\% in 2009 to 26.08\% in 2023 \citep{kojima2009perovskites,park2023perovskite}. While an immense number of compounds exhibit the perovskite crystal structure, not all are useful for solar applications. As experimentally determining the relevant properties of a perovskite compound is resource-intensive, there is significant interest in developing methods to screen for compounds that are likely to have desirable qualities. One method is to use density functional theory (DFT), a method for estimating the physical properties of a compound \citep{hohenberg1964dft, kohn1965dft,jones2015dft}.  However, the computational intensity of DFT makes it scale poorly to large numbers of candidate perovskites, motivating research into the use of machine learning techniques to approximate DFT outputs as an initial screening step \citep{chenebuah2021perovskites}. 

In this case study, we focus on predicting the DFT-output \textit{formation energy} and \textit{energy bandgap} of a perovskite material using data made available by \cite{chenebuah2021perovskites}. Finding the formation energy of a material is useful as it is related to its stability, and determining the energy bandgap is useful as it has a direct relationship with PV efficiency. \cite{chenebuah2021perovskites} applied multiple machine learning algorithms to predict formation energy and energy bandgap using data from 1,453 perovskite materials gathered from the Materials Project database \citep{jain2013materialsproject}. Predictors include element-based features, stability features, and crystallographic features; we refer to \cite[Table~2]{chenebuah2021perovskites} for a full description. Their methods focus on point predictions of formation energy and bandgap; in this case study, we extend their results by forming both point and interval predictions based on quantile estimation.

Formally, let $Y_i$, $i = 1, \dots, N$ be the DFT formation energy or energy bandgap and $\bm{X}_{i}$ a set of 56 covariates (two of the original covariates from \cite{chenebuah2021perovskites} were removed as they were almost perfectly collinear with other covariates). We assume that the couples $(X_i, Y_i)$, $i = 1, \dots, N$ are drawn \iid\ from a probability law. Our goal is to estimate the conditional median and conditional $\alpha$-quantiles (with $\alpha \in \{0.025, 0.05, 0.1, 0.5, 0.9, 0.95, 0.975 \}$) of the outcome conditional on covariates.
 Four candidate algorithms for the Super Learner ensemble were included in the ensemble: generalized random forests (GRF), directional random forests (DRF), gradient boosting machines (GBM), quantile neural networks (QRNN, with 2 hidden layers), and quantile regression (see Table~\ref{tab:base_learners}).

\paragraph{Results}
The case study results are presented in Table~\ref{tab:perovskite-results}. For the formation energy and energy bandgap the QSL achieved the lowest cross-validated quantile risk for all quantiles. In addition, the cross-validated empirical coverage of the QSL prediction intervals were the closest to the nominal level. For the energy bandgap, the QSL had the lowest or tied for the lowest empirical risk for 6 out of the 7 estimated quantiles. The QSL 90\% and 95\% prediction intervals were the closest to having the optimal empirical coverage, although the 80\% prediction interval undercovered relative to other methods.

\begin{table}
    \centering
    \begin{tabular}{|llllllll|}
    \hline
    & \multicolumn{7}{c|}{$\mathrm{EmpRisk}(\alpha)$} \\
    Algorithm & $\alpha = 0.025$ & $\alpha = 0.05$ & $\alpha = 0.1$ & $\alpha = 0.5$ & $\alpha = 0.9$ & $\alpha = 0.95$ & $\alpha = 0.975$ \\
    \hline
    \multicolumn{8}{|l|}{\textit{Formation Energy}} \\
    DRF & 0.026 & 0.044 & 0.068 & 0.13 & 0.079 & 0.053 & 0.035\\
    GBM & 0.024 & 0.037 & 0.049 & 0.078 & 0.055 & 0.042 & 0.029\\
    GRF & 0.029 & 0.048 & 0.07 & 0.12 & 0.073 & 0.051 & 0.032\\
    QRNN & 0.02 & 0.024 & 0.035 & 0.079 & 0.036 & 0.026 & \textbf{0.015}\\
    QReg & 0.017 & 0.026 & 0.042 & 0.094 & 0.044 & 0.027 & 0.017\\
    QSL & \textbf{0.012} & \textbf{0.021} & \textbf{0.034} & \textbf{0.065} & \textbf{0.035} & \textbf{0.021} & \textbf{0.015}\\
    \multicolumn{8}{|l|}{} \\
    \multicolumn{8}{|l|}{\textit{Energy Bandgap}} \\
    DRF & \textbf{0.043} & \textbf{0.086} & 0.17 & 0.4 & 0.2 & 0.13 & 0.077\\
    GBM & 0.047 & 0.095 & 0.17 & \textbf{0.3} & 0.16 & 0.1 & 0.067\\
    GRF & \textbf{0.043} & \textbf{0.086} & 0.17 & 0.38 & 0.19 & 0.12 & 0.072\\
    QRNN & 0.051 & 0.089 & 0.17 & 0.39 & 0.18 & 0.12 & 0.082\\
    QReg & 0.045 & 0.091 & 0.18 & 0.46 & 0.2 & 0.12 & 0.063\\
    QSL & \textbf{0.043} & 0.088 & \textbf{0.16} & \textbf{0.3} & \textbf{0.14} & \textbf{0.089} & \textbf{0.054}\\
    \hline
    \end{tabular}
    \caption{\label{tab:perovskite-results} Empirical risk \eqref{eq:emp-risk-iid} results from the perovskite case study (see Section \ref{section:case-study-perovskite}). The lowest empirical risk for each task and quantile is bolded. In the case of a tie, all tied algorithms are bolded.}
\end{table}

\begin{table}
    \centering
    \begin{tabular}{|lccc|}
    \hline
    & \multicolumn{3}{c|}{$\mathrm{EmpCov}(\beta)$} \\
    Algorithm & (1 - $\beta) = 0.8$ & $(1 - \beta) = 0.9$ & $(1 - \beta) = 0.95$ \\
    \hline
    \multicolumn{4}{|l|}{\textit{Formation Energy}} \\
    DRF & 94.2\% & 97.8\% & 99\%\\
    GBM & 57.9\% & 71.6\% & 84.4\%\\
    GRF & 94.6\% & 97.4\% & 98.7\%\\
    QRNN & 70.5\% & 78.2\% & 82\%\\
    QReg & 76.1\% & 86\% & 89.5\%\\
    QSL & \textbf{79.9\%} & \textbf{90.7\%} & \textbf{96.4\%}\\
    \multicolumn{4}{|l|}{} \\
    \multicolumn{4}{|l|}{\textit{Energy Bandgap}} \\
    DRF & 96.5\% & 98.5\% & 99.4\%\\
    GBM & 59.9\% & 74.9\% & 83.7\%\\
    GRF & 95.6\% & 98.4\% & 99.5\%\\
    QRNN & 74.7\% & 82\% & 83.7\%\\
    QReg & \textbf{76.4\%} & 86.2\% & 88.4\%\\
    QSL & 68.8\% & \textbf{90.1\%} & \textbf{97.2\%}\\
    \hline
    \end{tabular}
    \caption{\label{tab:perovskite-results-coverage} Empirical coverage \eqref{eq:emp-cov-iid} results for $(1 - \beta)\times 100\%$ prediction intervals from the perovskite case study (see Section \ref{section:case-study-perovskite}). The empirical coverage for each task and $\beta$ value are in bold. In the case of a tie, all tied algorithms are bolded.}
\end{table}

\subsection{Post-processing solar irradiance forecats}
\label{section:case-study-forecasting}
Solar irradiance is one of the principal variables influencing photovoltaic power output \citep{ahmed2020pv}. Short-term forecasts of solar irradiance are used to predict solar output, which aids electrical grid integration of solar power \citep{lorenz2011pv}. As a case study, we apply online quantile Super Learning to generate point predictions and well-calibrated prediction intervals for global horizontal irradiance (GHI) one day in advance, following the case study and data made available by \cite{wang2022solar}. 

As solar irradiance is mainly determined by local meteorological conditions, an important input for GHI forecasts are the outputs of weather prediction models. Numerical Weather Prediction (NWP) is typically based on highly complex deterministic models that output forecasts for a set of meteorological variables over a grid covering a geographical region (or the entire world). Multiple versions of a model are run with slightly perturbed initial conditions to yield a range of plausible weather trajectories which, taken together, are referred to as a \textit{dynamic ensemble} \citep{du2019weatherensemble}. While dynamic ensemble predictions cover a range of possible scenarios, they are not necessarily well-calibrated in a probabilistic sense \citep{schulz2021nwp}. For example, in the solar forecasting context, the dispersion of GHI forecasts across the members of a dynamic ensemble may not accurately reflect the variability in the eventually observed GHI. This mismatch motivates post-processing the dynamic ensemble forecasts with an algorithm that produces well-calibrated density estimates or prediction intervals based on the forecasts. If it is not known a-priori which algorithm will perform best for this task, as is almost always the case, then forming ensembles is warranted. It may also be advantageous to regularly update the parameters of the post-processing algorithm with new data as they become available, which leads naturally to the use of online Super Learning. 

For this case study, ground truth observations of GHI are taken from satellite measurements available from the National Solar Radiation Data Base (NSRDB; \cite{sengupta2018nsrdb}). As NWP input, we use archived forecasts produced by the European Centre for Medium-Range Weather Forecasts (ECMWF) Ensemble Prediction System. \cite{wang2022solar} released a subset of historical forecasts covering much of North America and Europe from 2017-2020 at horizons from 0h-90h in advance. Following their case study, and using the case study dataset they released, we focus on predicting GHI at 7 locations in the continental United States using one-day ahead ECMWF forecasts. 

Formally, let $Y_t$, $t = 1, \dots, T$ be the observed satellite measurement of GHI at 13:00 local time on day~$t$ at a single location. Let $X_{t,i}$  be the one-day ahead NWP forecasts of GHI covering the same location, where $i = 1, \dots, 50$ indexes the NWP ensemble members. We use as additional covariates the solar zenith angle $Z_t$ and the one-day lagged GHI observation $Y_{t-1}$. Our goal is to estimate conditional $10\%$, $50\%$, and $90\%$ quantiles of GHI separately for each location using quantile Super Learning. To ensure that each of the candidate learners have enough data to produce reasonable predictions, the online learning procedure is started on January 1, 2020, with data from 2017-2019 used as initial training data. Subsequently, each of the candidate algorithms is re-trained after each data point becomes available. The candidate algorithms are as follows (see Table~\ref{tab:base_learners} for the \texttt{R} packages we relied on):
\begin{itemize}
    \item Quantile regression 1: Quantile regression with covariates $X_{t,i}$, $i = 1, \dots, 50$.
    \item Quantile regression 2: Quantile regression with covariates $X_{t,i}$, $i = 1, \dots, 50$, $Z_t$, and $Y_{t-1}$. 
    \item GBM: Gradient Boosting Machines with covariates $X_{t,i}$, $i = 1, \dots, 50$, $Z_t$, and $Y_{t-1}$ trained with $500$ trees. 
    \item GRF: Generalized Random Forests with covariates $X_{t,i}$, $i = 1, \dots, 50$, $Z_t$, and $Y_{t-1}$.
    \item QGAM: Quantile Generalized Additive Models with covariates $X_{t,i}$, $i = 1, \dots, 50$, $Z_t$, and $Y_{t-1}$. Spline smooths were used to estimate the association between $Z_t$ and $Y_{t-1}$ and the outcome.
    \item QRNN: Quantile recurrent neural networks with covariates $X_{t,i}$, $i = 1, \dots, 50$, $Z_t$, and $Y_{t-1}$. 
\end{itemize}

The performance of the methods are compared by their empirical risk at time $T$ \eqref{eq:online-empirical-risk}. We also evaluate the empirical coverage of the prediction intervals formed from the 10\% and 90\% quantile forecasts of each method.

\paragraph{Results} The results of applying the quantile Super Learner to the case study dataset are shown in Table~\ref{tab:online-results-1}. The QSL had the lowest empirical risks in most cases for all but the 50\% quantile. For the 50\% quantile, the QSL tied or had slightly larger empirical risks than the EWA and BOA algorithms. All of the algorithms yielded prediction intervals that were close to the nominal level, with quantile Super Learner having the best-performing intervals (or tied for best) for five of the seven locations.

\begin{table}
    \centering
    \begin{tabular}{|ll|ccccccc|}
    \hline
    & & \multicolumn{7}{c|}{$\mathrm{EmpRisk}(\alpha)$ by location} \\
    $\alpha$-quantile & Method & BON & DRA & FPK & GWN & PSU & SXF & TBL \\
    \hline
    0.025 & BOA & 8.66 & 5.44 & \textbf{7.27} & \textbf{6.53} & 6.79 & 6.28 & 7.87\\
 & EWA & \textbf{7.83} & \textbf{5.38} & 7.82 & 6.91 & 7.34 & 6.77 & 8.04\\
 & QSL & 8.04 & 5.46 & 7.83 & 6.88 & \textbf{6.26} & \textbf{6.21} & \textbf{7.24}\\
\hline
0.05 & BOA & 12.7 & 7.82 & \textbf{11.7} & \textbf{10.9} & 10.2 & 10.9 & \textbf{12.1}\\
 & EWA & 12.7 & 7.97 & 12.2 & 11.3 & 11.9 & 11 & 12.3\\
 & QSL & \textbf{12.5} & \textbf{7.6} & 12.3 & \textbf{10.9} & \textbf{9.95} & \textbf{10.2} & 12.2\\
\hline
0.1 & BOA & 18.7 & 10.3 & 17.5 & 17.2 & 16.8 & 16.3 & \textbf{18.2}\\
 & EWA & 19 & 10.6 & 17.5 & 17.2 & 17.1 & 16 & \textbf{18.2}\\
 & QSL & \textbf{18.5} & \textbf{9.84} & \textbf{17} & \textbf{16.9} & \textbf{16.3} & \textbf{15.6} & 18.3\\
\hline
0.5 & BOA & \textbf{31.2} & \textbf{13.5} & 28.5 & 32.2 & \textbf{31.9} & \textbf{30.5} & \textbf{28}\\
 & EWA & 31.4 & \textbf{13.5} & 28.9 & 32.2 & 32.1 & 30.8 & 28.7\\
 & QSL & 31.3 & \textbf{13.5} & \textbf{28.4} & \textbf{32.1} & 32.1 & 30.9 & 28.4\\
\hline
0.9 & BOA & 15.2 & \textbf{5.34} & \textbf{12.5} & 14.3 & 14.2 & 13.2 & 11.7\\
 & EWA & 15.1 & 5.57 & 12.8 & 14.3 & 14.3 & 13.6 & 12.4\\
 & QSL & \textbf{14.8} & 5.47 & 12.6 & \textbf{13.9} & \textbf{13.9} & \textbf{12.9} & \textbf{11.2}\\
\hline
0.95 & BOA & 9.11 & 3.36 & 7.74 & 8.3 & 8.44 & 8.39 & 7.39\\
 & EWA & 9.36 & 3.11 & 8.12 & \textbf{8.14} & 8.16 & \textbf{8.17} & \textbf{7.01}\\
 & QSL & \textbf{8.9} & \textbf{3.06} & \textbf{7.72} & 8.19 & \textbf{8.1} & 8.39 & 7.29\\
\hline
0.975 & BOA & 5.4 & 2.39 & \textbf{4.58} & 5.79 & 4.54 & 5.4 & 4.29\\
 & EWA & \textbf{5.12} & 2.4 & 4.6 & 5.54 & 4.9 & \textbf{4.98} & 4.43\\
 & QSL & 5.98 & \textbf{2.14} & 4.96 & \textbf{5.11} & \textbf{4.38} & 5.52 & \textbf{4.25}\\
    \hline
    \end{tabular}
    \caption{\label{tab:online-results-1} Empirical risk of the quantile Super Learner (QSL), Exponentially Weighted Average (EWA), and Bernstein Online Aggregation (BOA) algorithms applied to point  forecasting the $\alpha$-quantiles of ground horizontal irradiance at seven locations in the continental United States (see Section \ref{section:case-study-forecasting}). The lowest empirical risks for each location are in bold. In the case of a tie, all tied algorithms are bolded. The locations are BON, Bondville, Illinois; DRA, Desert Rock, Nevada; FPK, Fort Peck, Montana; GWN, Goodwin Creek, Mississippi; PSU, Pennsylvania State University, Pennsylvania; SXF, Sioux Falls, South Dakota; and TBL, Table Mountain, Boulder, Colorado.}
\end{table}

\begin{table}
    \centering
    \begin{tabular}{|l|l|ccc|}
\hline
    & & \multicolumn{3}{c|}{$\mathrm{EmpCov}(\beta)$} \\
    Location & Method & $(1 - \beta) = 0.8$ & $(1 - \beta) = 0.9$ & $(1 - \beta) = 0.95$ \\
    \hline
    BON & BOA & 82.8\% & 91\% & \textbf{95.6\%}\\
 & EWA & 79.2\% & \textbf{89.6\%} & 94.3\%\\
 & QSL & \textbf{80.1\%} & 88.5\% & 93.2\%\\
\hline
DRA & BOA & 85.8\% & 91.8\% & \textbf{95.1\%}\\
 & EWA & 83.1\% & \textbf{90.7\%} & 95.1\%\\
 & QSL & \textbf{81.4\%} & 87.4\% & 95.1\%\\
\hline
FPK & BOA & 84.7\% & 91.3\% & \textbf{93.7\%}\\
 & EWA & 82.8\% & \textbf{89.9\%} & 92.3\%\\
 & QSL & \textbf{80.9\%} & 88.8\% & 90.7\%\\
\hline
GWN & BOA & 84.4\% & 91.3\% & 95.4\%\\
 & EWA & 82.8\% & \textbf{89.1\%} & 94.3\%\\
 & QSL & \textbf{81.7\%} & 87.7\% & \textbf{94.8\%}\\
\hline
PSU & BOA & 83.9\% & \textbf{91.3\%} & \textbf{93.7\%}\\
 & EWA & \textbf{80.3\%} & 88\% & 92.9\%\\
 & QSL & 77.3\% & 86.3\% & 91.5\%\\
\hline
SXF & BOA & 84.2\% & 91.5\% & \textbf{94.3\%}\\
 & EWA & 80.9\% & \textbf{89.6\%} & 92.3\%\\
 & QSL & \textbf{80.6\%} & 88.5\% & 91\%\\
\hline
TBL & BOA & 85.2\% & 91.8\% & \textbf{95.6\%}\\
 & EWA & \textbf{80.1\%} & \textbf{90.4\%} & 92.6\%\\
 & QSL & 80.9\% & 90.7\% & 92.6\%\\ 
 \hline
    \end{tabular}

    \caption{Empirical coverage of the $(1-\beta) \times 100\%$ prediction intervals of ground horizontal irradiance formed from quantile estimates based on quantile Super Learner (QSL), Exponentially Weighted Average (EWA), and Bernstein Online Aggregation (BOA) at seven locations in the continental United States (see Section \ref{section:case-study-forecasting}). The empirical coverage closest to the desired level for each location are in bold. In the case of a tie, all tied algorithms are bolded. The locations are BON, Bondville, Illinois; DRA, Desert Rock, Nevada; FPK, Fort Peck, Montana; GWN, Goodwin Creek, Mississippi; PSU, Pennsylvania State University, Pennsylvania; SXF, Sioux Falls, South Dakota; and TBL, Table Mountain, Boulder, Colorado.}
    \label{tab:my_label}
\end{table}

One way of understanding the empirical performance of each of the candidate algorithms is to examine how they are weighted in the quantile Super Learner ensemble. Figure \ref{fig:superlearner-solar-weights} shows the weights assigned to each candidate algorithm on the final day ($t = T$) for each of the seven locations.
Interestingly, the algorithms were weighted differently depending on the quantile being estimated. For the 10\% and 90\% quantiles, for example, gradient boosting machines received generally higher weights as compared to the 50\% quantile. In general, no single algorithm dominated across all locations and quantiles, illustrating the utility of ensemble based predictions.

\begin{figure}
    \centering
    \includegraphics[width=0.9\columnwidth]{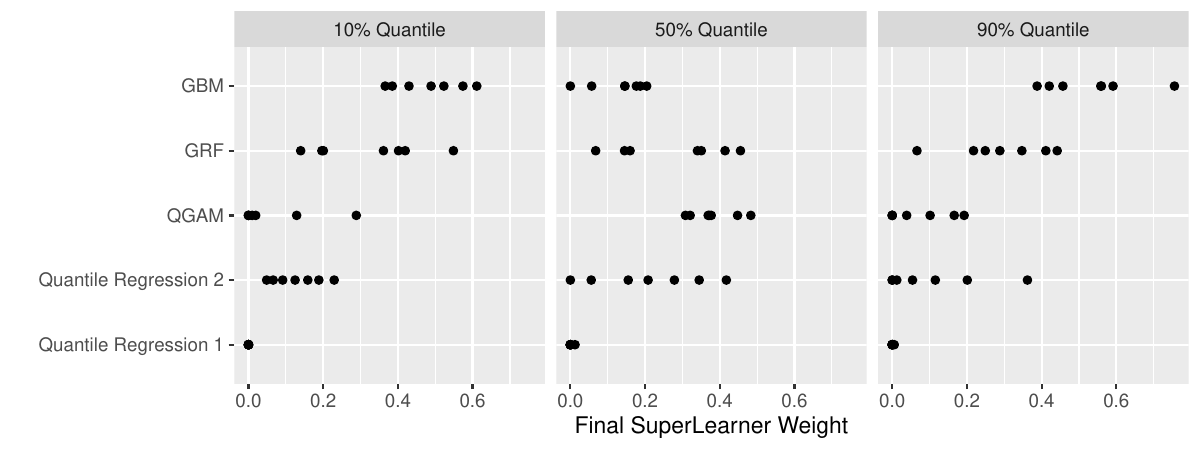}
    \caption{Final weights at time $t = T$ assigned to each of the candidate algorithms by the quantile Super Learner in the solar irradiance forecasting case study (see Section \ref{section:case-study-forecasting}).}
    \label{fig:superlearner-solar-weights}
\end{figure}

\section{Discussion}
\label{section:discussion}

We have proposed a method for conditional quantile estimation, Quantile Super Learning, that combines predictions from multiple candidate algorithms based on their performance measured with respect to a cross-validated empirical risk of the quantile loss function. The approach is theoretically grounded by excess risk bounds that hold with mild assumptions on the data generating distributions in both \iid\ and sequential data scenarios. 
 
Empirically, in simulation studies the QSL consistently achieved the lowest empirical quantile risk compared to the candidate algorithms in an \iid\ setup, showing the possible benefit of using ensemble methods. In the sequential setting, we found that online QSL outperformed Exponentially Weighted Average and Bernstein Online Aggregation algorithms in some settings. In the solar irradiance case study (see Section \ref{section:case-study-forecasting}), QSL tended to achieve lower empirical risk for all but the 50\% quantile. Practically QSL is also easy to use as it does not require specifying any tuning parameters, as opposed to EWA and BOA. However, the computational cost of QSL is significantly higher, as an optimization problem must be solved at each time step, as opposed to EWA and BOA in which the weights are updated by simple closed-form equations.

One possible use case for the QSL is to form prediction intervals by separately estimating lower and upper quantiles. However, in both the simulations and the case studies, while Super Learner consistently performed as well as or better than the candidate algorithms in terms of quantile risk, the prediction intervals formed via Super Learning did not always have the best empirical coverage. This reflects the fact that minimizing the quantile loss function does not necessarily lead to optimal coverage. Thus, Super Learner based prediction intervals based on minimizing quantile losses do not enjoy any performance guarantees. A natural extension of this work would be to post-process Super Learner prediction intervals using techniques from conformal inference, which have strong finite sample results. In the \iid\ case, conformalized quantile regression or the CV+ method could be used \citep{romano2019conformalized,barber2021cvplus}. \cite{fakoor2023qreg} provide a comprehensive empirical comparison of post-processing ensemble quantile estimators, and find conformal inference techniques performed well. In the online case,  Adaptive Conformal Inference techniques can be used to endow quantile based prediction intervals with finite sample coverage guarantees \citep{gibbs2021aci, gibbs2022conformal, zaffran2022agaci,bhatnagar2023saocp}. For all of these conformalization approaches it is advantageous to have good underlying estimates of the conditional quantile function, suggesting the use of ensemble methods to hedge against model misspecification.

We note that the goal and implementation of online Super Learning is similar to that of online aggregation of experts approaches in the online learning literature (see the comprehensive overview by \cite{cesalugosi2006onlinelearning}). Indeed, the online Super Learner functions identically to the Follow the Leader algorithm known to the online learning community. For both approaches, a convex combination of predictions of candidate algorithms is found that minimizes the empirical risk of the ensemble in hindsight. 
What differentiates them is their theoretical contexts and analyses. Follow the Leader is based on an online learning paradigm which makes no assumptions about how the observed data are generated, including even the possibility of data generated adversarially.
It is known to fail in such adversarial settings, and is not favored in the online learning community as there are other algorithms that have better worst-case properties 
\citep[Chapter 3.2]{cesalugosi2006onlinelearning}.
In the case of \iid\ data, a now standard online-to-batch argument can be used to translate results about the performance of online learning algorithms, such as regret bounds, to a statistical context. However, these arguments do not apply when the data are dependent, as is expected in many time-series settings.
In contrast, the typical analysis of Super Learning is based on a statistical point of view in which  the data are posited to follow a probability law for which we assume a stationarity condition, implying that the feature of interest of the conditional law can be learned. Crucially, this rules out the adversarial settings for which Follow the Leader is lacking.

\section*{Acknowledgements}
This research is partially supported by the Agence Nationale de la Recherche as part of the “Investissements d’avenir” program (reference ANR-19-P3IA-0001; PRAIRIE 3IA Institute).

\section*{Competing interests}
The authors declare no competing interests.

\bibliography{bibliography}

\appendix
\section{Theoretical Background}
In this appendix we present additional theoretical context and proofs of the theorems appearing in the main paper. 

\subsection{Independent setting}
Our theoretical treatment of the QSL in the \iid\ setting follows that of \cite{wu2022huber}, who use results from \cite{vandervaart2006oracles} and \cite{vanderLaan2007superlearner}. First, we require the definition of a pair of \textit{Bernstein numbers}, on which rests the theoretical analysis.
\begin{definition}[\cite{vandervaart2006oracles}]
Given a measurable $f : \mathcal{X} \times \mathbb{R} \to \mathbb{R}$, and for any $P \in \mathcal{M}$, call $(M(f), v(f))$ a $P$-pair of \textit{Bernstein numbers} of~$f$ if
\begin{align}
    M(f)^2 P\left( e^{|f| \slash M(f)} - 1 - \frac{|f|}{M(f)} \right) \leq \frac{1}{2} v(f).
\end{align}
\end{definition}
From the following result we see that, for $f$ uniformly bounded, the pair of Bernstein numbers is related to the supremum and variance of $f$:
\begin{lemma}[\cite{vandervaart2006oracles}]
\label{lemma:bernstein-uniformly-bounded}
If $f$ is uniformly bounded then, for any $P \in \mathcal{M}$, $\left(\|f\|_\infty, 1.5 Pf^2\right)$ is a $P$-pair of Bernstein numbers of $f$.
\end{lemma}

The following result establishes a Bernstein pair for the quantile loss function.
\begin{lemma}
    \label{lemma:bernstein-pair}
    Let $\psi \in \Psi$ be uniformly bounded. For any $P \in \mathcal{M}$ let $R^\alpha_P(\psi)$ be the risk of $\psi$ as defined in \eqref{eq:risk}. 
    Under assumptions A\ref{assumption:unique-quantiles-iid} and A\ref{assumption:boundedness}, the pair $(\bM, \bv)$ given by
    \begin{align}
        \bM &= \max\{ \alpha, 1 - \alpha \} (\| \psi \|_\infty + C_0) \quad \text{and} \quad
        \bv = 1.5 \times \bM \times R^\alpha_P(\psi)
    \end{align}
    is a $P$-pair of Bernstein numbers of $L^\alpha(\psi)$.
\end{lemma}
\begin{proof}
The proof follows closely that of \citet[Lemma 1]{wu2022huber}; therefore, we only summarize the proof by pointing out the relevant places where it differs. First, see that for all $o = (x, y) \in \mathcal{X} \times [-C_0, C_0]$,
\begin{align}
    0 \leq L^{\alpha}(\psi)(o) &= \alpha\left|y - \psi(x)\right| \mathbb{I}\left[y > \psi(x) \right] + \left(1-\alpha\right)\left|y-\psi(x)\right| \mathbb{I}\left[y\leq \psi(x)\right] \\
    &\leq \max\{\alpha, 1 - \alpha\} (\| \psi \|_\infty + C_0) =: \bM,
\end{align}
which is the first Bernstein number.
Next, we compute $PL^\alpha(\psi)^2$. Following \citet[Lemma 1]{wu2022huber}, we arrive at
\begin{align}
    P L^\alpha(\psi)^2 &= \bM \times R^\alpha_P(\psi) 
    =: \bv,
\end{align}
which is the second Bernstein number. This completes the proof.
\end{proof}

Next, we state an inequality bounding the difference between the cross-validated risk and the oracle risk.
\begin{theorem}[Theorem 2.3, \cite{vandervaart2006oracles}]
\label{theorem:general-inequality}
For any $P \in \mathcal{M}$, for any $\psi \in \Psi$ uniformly bounded, let $(\bM, \bv)$ be a $P$-pair of Bernstein numbers of the function $L^\alpha(\psi)$. Then for any $\delta > 0$ and $1 \leq p \leq 2$,
\begin{align}
    \mathbb{E}_{P_0}\left[\widetilde{R}^\alpha_{n,P}(\widehat{\psi}^\alpha_{\widehat{\kappa}_n})\right] \leq& (1 + 2\delta) \mathbb{E}_{P_0}\left[\widetilde{R}^\alpha_{n,P}(\widehat{\psi}^\alpha_{\widetilde{\kappa}_n})\right] \\
    &+ 16(1 + \delta) \log(1 + K) \\
    &\quad\times \mathbb{E}_{B_n}\left[ \sup_{\psi \in \Psi} \left( \frac{\bM}{n_1} + \left( \frac{\bv}{n^1 R^\alpha_P(\psi)^{2-p}} \right)^{1/p} \left( \frac{1 + \delta}{\delta} \right)^{2 / p - 1} \right)\right],
\end{align}
where $n^1 := \sum_{i=1}^n B_n(i)$. 
\end{theorem}

\paragraph{Proof of Theorem \ref{theorem:iid-oracle-inequality}.}
\begin{proof}
The proof follows that of \citet[Theorem 2]{wu2022huber},
with the substitution of $\max\{ 2\alpha C_0, 2(1-\alpha)C_0 \}$ for $C$. 
\end{proof}

\subsection{Online setting}
First, we present the following result that is key to the later analysis.
\begin{theorem}[Variance bound for the quantile loss \citep{steinwart2011}]
    \label{theorem:variance-bound-steinwart-christmann}
    Let $p \in (0, \infty]$, $q \in [0, \infty)$, and 
    \begin{align}
        \vartheta = \min\left\{ \frac{2}{q}, \frac{p}{p + 1} \right\}.
    \end{align}
    Let $Q$ be a distribution for $O_{j,\tau}$, $(j,\tau) \in \mathcal{J}\times\mathbb{N}^*$, that has an $\alpha$-quantile of $p$-average type $q$. For all $\psi \in \Psi$, define $\Delta L^\alpha(\psi) := L^\alpha(\psi) - L^\alpha(\psi^\alpha_{P_0})$. Then, for all $\psi \in \Psi$, it holds that
    \begin{align}
        Q\Delta L^\alpha(\psi)^2  \leq 2^{2-\vartheta} q^\vartheta \| \gamma^{-1} \|_{p, Q_{X_{j,\tau}}} \left( Q \Delta L_\alpha(\psi) \right)^\vartheta.
    \end{align}
\end{theorem}
In the above statement, $Q_{X_{j,\tau}}$ is the marginal law of $X_{j,\tau}$ under $Q$. Next, we present several necessary lemmas that follow from Assumptions B\ref{assumption:common-support}-B\ref{assumption:online-regularity}.
\begin{lemma}
    \label{lemma:online-boundedness}
    There exists $b_1 > 0$ such that $\sup_{\psi \in \Psi} \| \Delta L^\alpha(\psi)) \|_\infty < b_1$. In addition, there exists $b_2 \in (0, 2b_1]$ such that for all $j \in \mathcal{J}$, $t \geq 1$, and $\psi \in \Psi$, it holds $P_0$-almost surely that
    \begin{align}
        \left| \Delta L^\alpha(\psi)(O_{j,t}) - \mathbb{E}_{P_0}\left[ \Delta L^\alpha(\psi)(O_{j,t}) \mid F_{t-1} \right] \right| \leq b_2.
    \end{align}
\end{lemma}
\begin{proof}
    The existence of $b_1$ and $b_2$ follows directly from Assumption B\ref{assumption:online-boundedness}. 
\end{proof}

\begin{lemma}
    \label{lemma:online-variance-bound}
    There exist $\beta \in (0, 1]$ and $\nu > 0$ such that for all $j \in \mathcal{J}$, $t \geq 1$ and $\psi \in \Psi$, it holds $P_0$-almost surely that
    \begin{align}
        \mathbb{E}_{P_0}\left[ \left( \Delta L^\alpha(\psi)(O_{j,t}) \right) ^2 \middle| F_{t-1} \right] \leq \nu \left( \mathbb{E}_{P_0}[\Delta L^\alpha (\psi)(O_{j,t}) \mid F_{t-1}] \right)^\beta.
    \end{align}
\end{lemma}

\begin{proof}
By Assumption B\ref{assumption:online-regularity}, the $P_0$-conditional laws of $O_{j,t}$ given $F_{t-1}$ have $\alpha$-quantiles of $p$-average type $q$ for all $t \geq 1$. 
Let $\vartheta := \min\left\{ \frac{2}{q}, \frac{p}{p+1} \right \}$. By \citep[Theorem 2.8,][restated in the Appendix as Theorem~\ref{theorem:variance-bound-steinwart-christmann}]{steinwart2011}, 
\begin{align}
    \mathbb{E}_{P_0}\left[ \left( \Delta L^\alpha(\psi)(O_{j,t}) \right) ^2 \middle| F_{t-1} \right] \leq 2^{2-\vartheta} q^\vartheta \| \gamma_0^{-1} \|_{p}^\vartheta \left( \mathbb{E}_{P_0}[\Delta L^\alpha (\psi)(O_{j,t}) \middle| F_{t-1}] \right)^\vartheta,
\end{align}
where $\gamma_0$ is defined as $\gamma$ in Definition \ref{def:quantile-p-average-type-q} with the choice $Q$ equal to the conditional law of $O_{j,t}$ given $F_{t - 1}$, and $\|\gamma_0^{-1}\|_p$ is the $p$-norm of $\gamma_0^{-1}$ with respect to the marginal law of $X_{j,t}$ under $Q$. 
In view of the definition of $\Gamma$ in Assumption B\ref{assumption:online-regularity}, setting $\nu = 2^{2-\theta} q^\vartheta \Gamma^\vartheta > 0$ and $\beta = \vartheta$ and noting that $0 < \beta < 1$ complete the proof.
\end{proof}

\begin{lemma}
    \label{lemma:online-variance-bound-2}
    There exists $v_1 > 0$ such that, for all $j \in \mathcal{J}$, $t \geq 1$, and $\psi \in \Psi$, it holds $P_0$-almost surely that
    \begin{align}
        \mathbb{V}\mathrm{ar}[\Delta L^\alpha(\psi)(O_{j,t}) \mid F_{t-1}] \leq v_1.
    \end{align}
\end{lemma}
\begin{proof}
    The result follows from Lemma \ref{lemma:online-boundedness}.
\end{proof}

\begin{theorem}[Oracle Inequality for Online Super Learning \citep{ecoto2021onlinesuperlearner}]
    \label{theorem:online-oracle-inequality}
    Define
    \begin{align}
        v_2 := \frac{3\pi}{2} \left[ \left( \frac{15b_2}{|\mathcal{J}|} \right)^2 + \frac{64v_1}{|\mathcal{J}|} \right].
    \end{align}
    For any $\delta \in (0, 1]$, it holds that
    \begin{align}
        &\mathbb{E}_{P_0}\left[ \widetilde{R}^\alpha_{t,P_0}(\widehat{\psi}^\alpha_{\hat{\kappa}_t}) - \widetilde{R}_{t,P_0}^\alpha(\psi_{P_0}^\alpha) - (1 + 2\delta) \left( \widetilde{R}^\alpha_{t, P_0}(\widehat{\psi}^\alpha_{\widetilde{\kappa}_t}) - \widetilde{R}_{t,P_0}^\alpha(\widehat{\psi}_{P_0}^\alpha) \right)  \right] \\
        &\leq 3 \left( \frac{C_1(\delta)}{t} \log(2KN) \right)^{1/(2-\beta)} + \frac{2 C_2(\delta)}{t} \log(2KN),
    \end{align}
    where $C_1(\delta) := 2^{5-\beta}(1+\delta)^2 \gamma / \delta^\beta$, $C_2(\delta) := 8(1+\delta) b_2/3$, and $N \geq 2$ is chosen such that
    \begin{align}
        N \geq \frac{\beta}{2-\beta} \frac{\log(t) + \log(C_3)}{\log(2)},
    \end{align}
    with $C_3 := (v_2 / \gamma)^{(2 - \beta) / \beta} / (2^{5-\beta} \gamma)$.
\end{theorem}

\begin{proof}
    It suffices to check each of the 5 conditions of \cite[Corollary 2]{ecoto2021onlinesuperlearner}.
    \begin{itemize}
        \item Assumption 1 is satisfied by our Assumption B\ref{assumption:online-markov}.
        \item Assumption 2 is satisfied by our Assumption B\ref{assumption:stationarity}.
        \item Assumption 3 is satisfied according to our Lemma \ref{lemma:online-boundedness}.
        \item Assumption 4 is satisfied according to our Lemma \ref{lemma:online-variance-bound}.
        \item Assumption 5 is satisfied because Assumption 3 is satisfied.
    \end{itemize}
\end{proof}

\paragraph{Proof of Theorem \ref{theorem:online-main-result}}
\begin{proof}
To simplify the exposition, let $\mathrm{ER}_{\tilde{\kappa}} := \widetilde{R}^\alpha_{t, P_0}(\widehat{\psi}^\alpha_{\widetilde{\kappa}_t}) - \widetilde{R}_{t,P_0}^\alpha(\widehat{\psi}_{P_0}^\alpha)$. By Theorem \ref{theorem:online-oracle-inequality},
\begin{align}
        \mathbb{E}_{P_0}\left[ \widetilde{R}^\alpha_{t,P_0}(\widehat{\psi}^\alpha_{\hat{\kappa}_t}) - \widetilde{R}_{t,P_0}^\alpha(\psi_{P_0}^\alpha)  \right] 
        \leq& \mathbb{E}_{P_0}[\mathrm{ER}_{\tilde{\kappa}}] + 2\delta \mathbb{E}_{P_0}[\mathrm{ER}_{\tilde{\kappa}}] + \mathrm{Rem}(\delta)
    \end{align}
    where
    \begin{align}
        \mathrm{Rem}(\delta) := 3\left( \frac{C_1(\delta)}{t} \log(2KN) \right)^{\frac{1}{2-\beta}} + \frac{2 C_2(\delta)}{t} \log(2KN).
    \end{align}
    Choose an integer $N \in [ 3\log(t), 4\log(t) ]$, which will necessarily satisfy
    \begin{align}
        N \geq \frac{\beta}{2 - \beta} \frac{\log(t) + \log(C_3)}{\log(2)}
    \end{align}
    provided that $t \geq \max\{2, C_3\}$.
    Fix $\delta = t^{-1/2}$.
    Then 
    \begin{align}
        C_2(\delta) \leq 16b_2/3 \text{ and } 2^{5-\beta} \gamma t^{\frac{\beta}{2}} \leq C_1(\delta) \leq 2^{7-\beta} \gamma t^{\frac{\beta}{2}}.
    \end{align} 
    Therefore, for $t$ large enough,
    \begin{align}
        \mathrm{Rem}(\delta) &\lesssim C_1(t^{-1/2})^{\frac{1}{2-\beta}} \left[ \left(\frac{\log(2KN}{t}\right)^{\frac{1}{2-\beta}} + \frac{\log(2KN)}{t} \right] \\
        &\lesssim C_1(t^{-1/2})^{\frac{1}{2-\beta}} \left[ \left(\frac{\log(8K\log(t)}{t}\right)^{\frac{1}{2-\beta}} + \frac{\log(8K\log(t))}{t} \right] \\
        &\lesssim \left[ C_1(t^{-1/2}) \frac{\log(8K\log(t))}{t}  \right]^{\frac{1}{2-\beta}} \\
        &\lesssim \left[\frac{\log(8K\log(t))}{t^{1-\frac{\beta}{2}}}  \right]^{\frac{1}{2-\beta}} \\
        &\lesssim \frac{\log(8K\log(t))}{t^{1/2}}.
    \end{align}
    Next, note that for large enough $t$ and using the boundedness assumption,
    \begin{align}
        2\delta \mathbb{E}[\mathrm{ER}_{\tilde{\kappa}}] 
         &= 2 t^{-1/2} \mathbb{E}[\mathrm{ER}_{\tilde{\kappa}}] \leq 2t^{-1/2} C_0 \lesssim \frac{\log(8K\log(t))}{t^{1/2}}.
    \end{align}
    Therefore \begin{align}
        2\delta \mathbb{E}[\mathrm{ER}_{\tilde{\kappa}}] + \mathrm{Rem}(\delta) = O\left( \frac{\log(8K\log(t))}{t^{1/2}} \right),
    \end{align}
    which completes the proof.
\end{proof}


\end{document}